\documentclass[letterpaper, 10 pt, conference]{ieeeconf}  

\IEEEoverridecommandlockouts                              
\usepackage{color} 
\usepackage{bm}

\usepackage{amsmath,amsthm} 
\usepackage{graphics} 
\usepackage{epsfig} 
\usepackage{times} 
\usepackage{amssymb}  
\usepackage{mathtools}
\usepackage{caption}
\usepackage{subcaption}

\theoremstyle{definition}
\newtheorem{problem}{Problem}
\newtheorem{example}{Example}
\newtheorem{assumption}{Assumption}
\newtheorem{definition}{Definition}

\newtheorem{remark}{Remark}
\newtheorem{theorem}{Theorem}
\newtheorem{corollary}{Corollary}
\newtheorem{case study}{Case Study}
\usepackage[linesnumbered,ruled,vlined]{algorithm2e}

\title{\LARGE \bf
Learning Robust and Correct Controllers from \\ Signal Temporal Logic Specifications Using BarrierNet}

\author{Wenliang Liu$^{1}$, Wei Xiao$^{2}$, and Calin Belta$^{1}$
\thanks{*This work was partially supported by the National Science Foundation under grant IIS-2024606 at Boston University.}
\thanks{$^{1}$Wenliang Liu and Calin Belta are with Department of Mechanical Engineering,
        Boston University, MA, USA
        {\tt\small wliu97@bu.edu, cbelta@bu.edu}}%
\thanks{$^{2}$Wei Xiao is with the Computer Science and Artificial Intelligence Lab, Massachusetts Institute of Technology, MA, USA.
        {\tt\small weixy@mit.edu}}%
}

\begin{document}

\maketitle
\thispagestyle{empty}
\pagestyle{empty}

\begin{abstract}
In this paper, we consider the problem of learning a neural network controller for a system required to satisfy a Signal Temporal Logic (STL) specification. We exploit STL quantitative semantics to define a notion of robust satisfaction. Guaranteeing the correctness of a neural network controller, i.e., ensuring the satisfaction of the specification by the controlled system, is a difficult problem that received a lot of attention recently. We provide a general procedure to construct a set of trainable High Order Control Barrier Functions (HOCBFs) enforcing the satisfaction of formulas in a fragment of STL. We use the BarrierNet, implemented by a differentiable Quadratic Program (dQP) with HOCBF constraints, as the last layer of the neural network controller, to guarantee the satisfaction of the STL formulas. We train the HOCBFs together with other neural network parameters to further improve the robustness of the controller. Simulation results demonstrate that our approach ensures satisfaction and outperforms existing algorithms. 
\end{abstract}

\section{INTRODUCTION}

Autonomous and robotic systems are usually tasked to satisfy requirements that go beyond stability and set invariance. For example, in a surveillance application, an autonomous aircraft may be required to gather data from a region of interest every 3 hours, charge at its docking station for at least 20 minutes every 2 hours, and avoid a no-flight zone for all times.   
Temporal logics, such as Linear Temporal Logic (LTL) \cite{baier2008principles} and Signal Temporal Logic (STL) \cite{maler2004monitoring}, have been widely used as specification languages due to their rich expressivity. 

In this paper, we consider the problem of controlling a system to satisfy a specification given as a STL formula. This logic is interpreted over real-valued signals and it has both qualitative semantics, in which a signal either satisfies or violates a formula, and quantitative semantics (also known as robustness) \cite{donze2010robust}, in which a signal is associated a real value that measures how strongly the specification is satisfied. It was shown that controlling a system required to satisfy an STL specification can be formulated as an optimization problem with robustness as objective or as a constraint, which can be solved using Mixed Integer Programming (MIP) \cite{raman2014model,sadraddini2015robust} or gradient-based optimization \cite{pant2017smooth,mehdipour2019arithmetic,gilpin2020smooth}. Such methods, however, are computationally expensive and difficult to use for online control.

Reinforcement Learning (RL)-based techniques can perform most of the computation offline, hence enabling real-time control. 
Model-based RL using neural network was applied to control synthesis problems under STL tasks in \cite{yaghoubi2019worst,liu2021safe,leung2022semi}, where the robustness was used as an objective (reward) function to learn a robust controller.  However, these works cannot guarantee the correctness of the learned policy, i.e., satisfaction of the specification by the system under the derived policy. Violation can have two main causes. First, while training a neural network, the system can get stuck at a local optimum, which can be far from the global optimum. This can result in a policy leading to unwanted behavior both during training and testing, and it is likely to happen when the STL specification and the system dynamics are complex. Second, even if the neural network converges to a policy that satisfies the STL specification during training, when given unseen initial conditions or environments in testing, the policy can still fail. The works in  \cite{yaghoubi2019worst} and \cite{leung2022semi} use falsification methods, while \cite{liu2021safe} uses Control Barrier Functions (CBF) to mitigate the second problem, but none of them can guarantee satisfaction. Q-learning is also considered for STL control synthesis in \cite{aksaray2016q,venkataraman2020tractable}. This provides no guarantee of satisfaction either. The authors of \cite{kalagarla2021model} use constrained Markov Decision Process (cMDP) to provide a lower bound on the probability of satisfying an STL specification. 

In this paper, we use model-based RL and assume that the model (system dynamics) is known. We propose an algorithm to learn a control policy that is guaranteed to satisfy the given STL specification during both training and deployment by using CBFs.
These types of functions have been extensively used in the controls community to guarantee safety specified as set invariance \cite{ames2016control,ames2019control}. CBFs have also been employed to enforce the satisfaction of STL specifications. The authors of \cite{lindemann2018control} used time-varying CBFs to satisfy tasks given in a fragment of STL. The controller was obtained via a quadratic program (QP), which can be solved efficiently. In \cite{xiao2021high}, high order control Lyapunov-barrier functions were defined and used to satisfy STL tasks for systems with arbitrary relative degrees. The methods in \cite{lindemann2018control,xiao2021high} require manual design of the CBFs corresponding to the STL and the parameters in the constraints. A bad design may result in increased conservativeness or even infeasibility. Recently, we proposed BarrierNet \cite{xiao2021barriernet}, implemented as a differentiable QP with CBF constraints, as the last layer of a neural network controller to guarantee safety. In this method, the parameters in the CBF constraints can be obtained through training, which results in significant decrease in conservativeness.  

In this paper, we combine BarrierNet \cite{xiao2021barriernet} with time-varying CBFs for STL tasks \cite{lindemann2018control} to train a neural network controller that guarantees the satisfaction of formulas in a fragment of STL that contains no nested temporal operators and the ``\emph{until}" operator.
We extend \cite{lindemann2018control} to High Order Control Barrier Functions (HOCBFs) \cite{xiao2019control} and provide a general, algorithmic procedure to generate these HOCBFs given an STL formula. Further, unlike the fixed CBFs in \cite{lindemann2018control}, our HOCBFs contain parameters that can be trained together with the neural network controller using BarrierNet. As a result, our approach avoids the complicated manual design in \cite{lindemann2018control} and reduces the conservativeness after training.  Our results show that the learned policy achieves a higher robustness than directly applying CBFs as in \cite{lindemann2018control}. 
Unlike \cite{xiao2021barriernet} where the policy is trained on a dataset using supervised learning, we apply model-based RL to train the policy as in \cite{liu2021safe}. Therefore, no dataset is needed during training. The trained controller can be implemented in real-time and generalized to random initial conditions while retaining correctness. 


\section{Preliminaries}
\label{sec:prelim}

We use non-bold letters $x$, bold letters $\mathbf x$, and calligraphic letters $\mathcal X$ to denote scalars, vectors, and sets respectively. Consider a nonlinear control-affine system:
\begin{equation}
    \label{eq:system}
    \mathbf{\dot x} = f(\mathbf x) + g(\mathbf x)\mathbf u,
\end{equation}
where $\mathbf x\in\mathbb R^n$ is the system state, $\mathbf u \in\mathcal U\in \mathbb R^q$ is the control, $f:\mathbb R^n\rightarrow \mathbb R^n$ and $g:\mathbb R^n\rightarrow \mathbb R^{n\times q}$ are locally Lipschitz continuous functions. We assume $\mathcal U$ is a box constraint, i.e., $\mathbf u_{min}\leq\mathbf u\leq\mathbf u_{max}$, where the inequality is interpreted element-wise. Without loss of generality, we assume the initial time is $0$. The initial condition $\mathbf x(0)=\mathbf x_0$ is randomly sampled in a set $\mathcal X_0\in \mathbb R^n$ with probability density function $P:\mathcal X_0\rightarrow \mathbb R$. We consider solutions to \eqref{eq:system} over a compact time interval $[0,T]$. Given an initial condition $\mathbf x_0\in\mathcal X_0$ and a control signal $\mathbf u:[0,T]\rightarrow \mathcal U$, a signal $\mathbf x:[0,T] \rightarrow \mathbb R^n$ is a solution of \eqref{eq:system} if $\mathbf x(t)$ is absolutely continuous and satisfies \eqref{eq:system} for all $t\in [0,T]$. A partial solution on $[0,t]$ is denoted as $\mathbf x_{0:t}: [0,t]\rightarrow \mathbb R^n$. We define a state-feedback neural network controller with memory as 
\begin{equation}
\label{eq:nn}
    \mathbf u(t) = \pi(\mathbf x_{0:t},\bm\theta),
\end{equation}
where $\bm\theta$ is a set of neural network parameters. Memory can be enabled by using Recurrent Neural Network (RNN) \cite{goodfellow2016deep}. 

\subsection{Signal Temporal Logic (STL)}

Signal Temporal Logic \cite{maler2004monitoring} is interpreted over real-valued signals $\mathbf x:\mathbb R_{\geq0}\rightarrow \mathbb R^n$, e.g., solutions of \eqref{eq:system}. In this paper, we consider a fragment of STL with the following syntax:
\begin{subequations}
\label{eq:stl}
\begin{align}
\label{eq:stl1}
    \phi &\coloneqq \top\ |\ \mu\ |\ \neg \mu\ |\ \phi_1\land\phi_2 \\
\label{eq:stl2}
    \varphi &\coloneqq F_{[t_a,t_b]}\phi\ |\ G_{[t_a,t_b]}\phi\ |\  \varphi_1\land\varphi_2,
\end{align}
\end{subequations}
where $\phi$ and $\varphi$ are STL formulae, $\phi_1$ and $\phi_2$ are formulae of class $\phi$ while $\varphi_1$, $\varphi_2$ are formulae of class $\varphi$, $\top$ is the logical \emph{true}, $\mu$ is a predicate in the form of $h(\mathbf x)\geq 0$ with $h:\mathbb R^n\rightarrow \mathbb R$, $\neg$ and $\land$ are Boolean negation and conjunction respectively, $F$ and $G$ are temporal \emph{eventually} and \emph{always} respectively, $[t_a,t_b]$ is a time interval with $t_a<t_b$. 

We use $(\mathbf x, t)\models \varphi$ to denote that signal $\mathbf x$ satisfies $\varphi$ at time $t$. A formal definition of qualitative semantics of STL can be found in \cite{maler2004monitoring}. Informally, $F_{[t_a,t_b]}\phi$ is satisfied if ``$\phi$ becomes \textit{True} at some time in $[t_a,t_b]$" while $G_{[t_a,t_b]}\phi$ is satisfied if ``$\phi$ is \textit{True} at all time in $[t_a,t_b]$". Other Boolean operators are interpreted in the usual way. Compared with the full STL \cite{maler2004monitoring}, the STL fragment \eqref{eq:stl} cannot contain the temporal \emph{until} or nested temporal operators like ``eventually always". However, it is still capable of expressing a wide range of useful temporal properties in practice, e.g., safety and reachability constraints with concrete time requirements.

STL is also equipped with quantitative semantics, also called robustness, which is a real value that measures how much a signal satisfies $\varphi$. Multiple STL robustness measures have been proposed \cite{donze2010robust,mehdipour2019arithmetic,gilpin2020smooth,varnai2020robustness}. In this paper, we use the smooth robustness defined in \cite{liu2022robust}, which is differentiable almost everywhere, and easy to embedded in learning-based algorithms. The robustness is sound in the sense that the robustness value is positive if and only if the STL formula is satisfied. We denote the robustness of $\varphi$ at time $t$ with respect to a signal $\mathbf x$ as $\rho(\varphi,\mathbf x, t)$. Further, we define the time horizon of an STL formula $\varphi$ as $hrz(\varphi)$, which is the closest time point in the future that is required to determine the satisfaction and robustness of $\varphi$. In this paper, we only consider the solution of system \eqref{eq:system} within the time horizon of the given STL formula, i.e., $T = hrz(\varphi)$. 

\subsection{Time-Varying High Order Control Barrier Function}

In this subsection we introduce time-varying High Order Control Barrier Function (HOCBF) \cite{xiao2019control}. We start with giving the definition of class $\mathcal K$ function:
\begin{definition}
    (Class $\mathcal K$ function) A continuous function $\alpha:[0,a)\rightarrow[0,\infty)$ belongs to class $\mathcal K$ if it is strictly increasing and $\alpha(0)=0$.
\end{definition}

Informally, the relative degree of a (sufficiently many times) differentiable time-varying function $b:\mathbb R^n\times[0,T]\rightarrow\mathbb R$ defined over the state of system \eqref{eq:system} is the number of times it needs to be differentiated along its dynamics until all elements in the control $\mathbf u$ show up. 
Consider a constraint $b(\mathbf x,t)\geq 0$ where $b:\mathbb R^n\times[0,T]\rightarrow\mathbb R$ is a differentiable function with relative degree $m$. Let $\psi_0(\mathbf x,t) \coloneqq b(\mathbf x,t)$. We define a sequence of functions $\psi_i:\mathbb R^n\times[0,T]\rightarrow \mathbb R$, $i=1,\ldots,m$ as follows:
\begin{equation}
\label{eq:psi}
    \psi_i(\mathbf x,t) \coloneqq \dot \psi_{i-1}(\mathbf x,t) + \alpha_i\big(\psi_{i-1}(\mathbf x,t)\big),
\end{equation}
where $\alpha_i$, $i=1,\ldots,m$ is a $(m-i)^{th}$ order differentiable class $\mathcal K$ function. Let $\mathcal C_i(t)$ be the super-level set of $\psi_i(\mathbf x,t)$:
\begin{equation}
\label{eq:c}
    \mathcal C_i(t) = \{\mathbf x\in \mathbb R^n|\psi_i(\mathbf x,t)\geq 0\}.
\end{equation}
\begin{definition}
    (HOCBF \cite{xiao2019control}) Let $\psi_1(\mathbf x,t),\ldots,\psi_m(\mathbf x,t)$ be defined by \eqref{eq:psi} and $\mathcal C_1(t),\ldots,C_m(t)$ be defined by \eqref{eq:c}. A differentiable function $b(\mathbf x,t)$ is a High Order Control Barrier Function (HOCBF) with relative degree $m$ with respect to system \eqref{eq:system} if there exist differentiable class $\mathcal K$ functions $\alpha_i$, $i=1,\ldots,m$, such that
    \begin{equation}
    \label{eq:hocbf}
    \begin{aligned}
        \sup_{\mathbf u\in\mathcal U}\big[&L_f^m b(\mathbf x,t) + L_gL_f^{m-1}b(\mathbf x,t)\mathbf u + \frac{\partial^mb(\mathbf x,t)}{\partial t^m}\\
        &+O(b(\mathbf x,t)) + \alpha_m(\psi_{m-1}(\mathbf x,t))\big] \geq 0,
    \end{aligned}
    \end{equation}
    for all $(\mathbf x,t)\in \mathcal C_1(t)\cap\mathcal C_2(t)\cap\ldots\cap\mathcal C_m(t)\times[0,T]$. In \eqref{eq:hocbf}, $L_f^m$ ($L_g$) denotes Lie derivatives along $f$ ($g$) $m$ (one) times, and $O(b(\mathbf x,t))$ denotes the remaining Lie derivatives along $f$ and partial derivatives with respect to $t$ with degree less than $m$. 
\end{definition}
Note that HOCBF is a general form of CBF \cite{ames2019control}. Setting $m=1$ reduces the HOCBF to the common form of CBF.
\begin{definition}
    (Forward invariant) A set $\mathcal C(t)\subset\mathbb R^n$ that depends on time is forward invariant for system \eqref{eq:system} given a control law $\mathbf u$ if for any $\mathbf x(0)\in\mathcal C(0)$, the solution of system \eqref{eq:system} satisfies $\mathbf x(t)\in\mathcal C(t)$, $\forall t\in[0,T]$. 
\end{definition}

\begin{theorem}
\label{thm:cbf}
\cite{xiao2019control} Given an HOCBF $b(\mathbf x,t)$ with a sequence of sets $\mathcal C_1(t),\ldots,C_m(t)$ as defined in \eqref{eq:c}, if $\mathbf x(0)\in \mathcal C_1(0)\cap\mathcal C_2(0)\cap\ldots\cap\mathcal C_m(0)$, then any Lipschitz continuous controller $\mathbf u(t)$ that satisfies \eqref{eq:hocbf} $\forall t\in[0,T]$ renders $\mathcal C_1(t)\cap\mathcal C_2(t)\cap\ldots\cap\mathcal C_m(t)$ forward invariant for system \eqref{eq:system}.
\end{theorem}

\section{Problem Formulation and Approach}
\label{sec:prob-form}

Let $J(\mathbf u)$ be a cost function over control signals $\mathbf u:[0,T]\rightarrow\mathcal U$. The problem we consider in this paper is:

\begin{problem}
\label{pb:1}
Given a system with known dynamics \eqref{eq:system}, an STL specification $\varphi$ as in \eqref{eq:stl}, and a initial state $\mathbf x_0$ sampled from the distribution $P:\mathcal X_0\rightarrow \mathbb R$, find the optimal control $\mathbf u^*(t)$ 
that maximizes the STL robustness and minimize the cost $J(\mathbf u)$ while guaranteeing the satisfaction of $\varphi$:
\begin{equation}
    \label{eq:goal}
    \begin{aligned}
       \mathbf u^*(t) = &\arg\max_{\mathbf u(t)} \rho(\varphi,\mathbf x, 0) - J(\mathbf u) \\
       \text{s.t.}\quad & \dot{\mathbf x} = f(\mathbf x) + g(\mathbf x)\mathbf u(t), \\
       & \mathbf u_{min} \leq \mathbf u(t) \leq \mathbf u_{max},\\
       & (\mathbf x,0)\models \varphi.
    \end{aligned}
\end{equation}
\end{problem}

To be robust against disturbances, a feedback controller is desired. One can obtain such a feedback controller by solving \eqref{eq:goal} at each discrete time step in a model predictive control manner as in \cite{raman2014model,sadraddini2015robust}. However, doing so can be time-consuming and prevent real-time control. Training a neural network controller that maximizes the expected objective in \eqref{eq:goal} over initial state distribution $P$ can move the online computation to offline. After training, the controller can be computed in real-time and can be generalized to random initial conditions under the distribution $P$ \cite{liu2021safe}. Moreover, in general, an STL specification is history-dependent \cite{liu2021recurrent}, i.e., to satisfy it, the desired control $\mathbf u(t)$ should depend on not only the current state $\mathbf x(t)$ but also history states $\mathbf x_{0:t}$. Hence, a controller with memory is needed. 

In this paper, we train a neural network controller with memory \eqref{eq:nn} to solve Problem \ref{pb:1}. 
We first construct a set of trainable time-varying HOCBFs from the STL formula $\varphi$. Then we embedded these HOCBFs into the neural network controller using a modified version of the BarrierNet from \cite{xiao2021barriernet}
to guarantee the satisfaction of $\varphi$. We train the neural network controller together with the HOCBFs to further increase the STL robustness. 

\section{Solution}
\label{sec:solution}

In this section, we present our solution for Problem \ref{pb:1}. We first introduce the trainable HOCBF and a modified version of BarrierNet from \cite{xiao2021barriernet} in Section \ref{subsec:barriernet}. Then we design a general procedure to construct a set of time-varying HOCBFs that can be used to guarantee the satisfaction of a given STL specification in Section \ref{subsec:hocbf}. Then we describe how these time-varying HOCBFs are trained together with the neural network controller using BarrierNet in Section \ref{subsec:learning} to further improve robustness.

\subsection{Trainable HOCBF and BarrierNet}
\label{subsec:barriernet}
Suppose that we have a set of time-varying HOCBFs $b_j(\mathbf x,t,\bm\theta_b,\mathbf x_0)$ that depend on the initial condition $\mathbf x_0$ and contain trainable parameters $\bm\theta_b$, $j=1,\ldots,M$. The reason they depend on $\mathbf x_0$ will be clear in Section \ref{subsec:learning}. To avoid over-conservativeness, we make the class $\mathcal K$ functions also trainable. Rewrite \eqref{eq:psi} for a HOCBF $b_j$ into:
\begin{equation}
    \label{eq:psi'}
    \psi_{i,j}(\mathbf x,t) \coloneqq \dot \psi_{i-1,j}(\mathbf x,t) + p_{i,j}(\mathbf x_0,\bm\theta_p)\alpha_{i,j}\big(\psi_{i-1,j}(\mathbf x,t)\big),
\end{equation}
where $\alpha_{i,j}$ are given class $\mathcal K$ functions, $p_{i,j}(\mathbf x_0,\bm\theta_p)>0$, $i=1,\ldots,m_j$, $j=1,\ldots,M$, $m_j$ is the relative degree of HOCBF $b_j$. $p_{i,j}(\mathbf x_0,\bm\theta_p)$ also depends on initial condition and contains trainable parameters $\bm\theta_p$. The reason it depends on $\mathbf x_0$ will be clear in Section \ref{subsec:learning} as well.

BarrierNet \cite{xiao2021barriernet} is a neural network layer implemented by a differentiable Quadratic Program (dQP) with HOCBF constraints. We add it as the last layer of a neural network controller (with memory) \eqref{eq:nn}, i.e., $\pi(\mathbf x_{0:t},\bm\theta)=\mathbf u^*(t)$ with $\mathbf u^*(t)$ given by:
\begin{equation}
    \label{eq:barriernet}
    \begin{aligned}
    &\mathbf u^*(t) \hspace{-2pt} = & & \hspace{-10pt} \arg\min_{\mathbf u(t)}\frac{1}{2}\mathbf  u(t)^\top \mathbf Q(\mathbf x_{0:t},\bm\theta_q) \mathbf u(t) + \mathbf F^\top(\mathbf x_{0:t},\bm\theta_f) \mathbf u(t)    \\
    & \text{s.t.}& & \hspace{-10pt} L_f^m b_j(\mathbf x,t,\bm\theta_b,\mathbf x_0) + L_gL_f^{m-1}b_j(\mathbf x,t,\bm\theta_b,\mathbf x_0)\mathbf u(t)\\ 
    & & & \hspace{-10pt} + \frac{\partial^mb_j(\mathbf x,t,\bm\theta_b,\mathbf x_0)}{\partial t^m} +O(b_j(\mathbf x,t,\bm\theta_b,\mathbf x_0))\\ 
    & & & \hspace{-10pt} + p_{m,j}(\mathbf x_{0},\bm\theta_p)\alpha_m(\psi_{m-1,j}(\mathbf x,t,\bm\theta_b,\mathbf x_0))\geq 0,\\
    & & &\hspace{-10pt} t=k\Delta t,\ k=0,1,2,\ldots,\ j=1,\ldots,M,
    \end{aligned}
\end{equation}
where $\mathbf Q(\mathbf x_{0:t},\bm\theta_q)\in\mathbb R^{q\times q}$, $\mathbf F(\mathbf x_{0:t},\bm\theta_f)\in\mathbb R^q$, $p_{i,j}(\mathbf x_{0},\theta_p)$, $i=1,\ldots,m$ and $b_j(\mathbf x,t,\bm\theta_b,\mathbf x_0)$ are all given by previous neural network layers with trainable parameters $(\bm\theta_q,\bm\theta_f,\bm\theta_p,\bm\theta_b)\coloneqq\bm\theta$, $\mathbf Q$ is positive definite. $\bm Q^{-1}\mathbf F$ can be interpreted as a reference control. Although in \eqref{eq:barriernet}, $\mathbf Q$, $\mathbf F$, $p_i$ are given by previous layers, they can also be directly trainable parameters. 
The dQP \eqref{eq:barriernet} is solved at each time point $k\Delta t$, $k=0,1,\ldots$ until reaching the time horizon $T$, and the solution $\mathbf u^*(t)$ is applied to the system as a constant for the time period $[k\Delta t, k\Delta t + \Delta t)$. Since \eqref{eq:barriernet} is differentiable, the gradient of $\mathbf u^*(t)$ with respect to $\bm\theta$ can be calculated using the technique in \cite{amos2017optnet}, then $\bm\theta$ can be trained using any methods for training neural networks. Different from the original BarrierNet \cite{xiao2021barriernet}, in \eqref{eq:barriernet} we also make the HOCBF $b$ itself trainable besides $\mathbf Q$, $\mathbf F$ and $p_i$, as it will be detailed in the next subsection. BarrierNet is able to guarantee the satisfaction of all HOCBF constraints. Meanwhile, through training the controller can also optimize a given objective function. 

\subsection{HOCBFs for STL specifications}
\label{subsec:hocbf}
The authors of \cite{lindemann2018control} proposed the idea of using time-varying CBF to ensure the satisfaction of a given STL specification. However, in \cite{lindemann2018control} only relative degree $1$ CBFs are considered and the generation of CBFs is described by examples without explicitly showing the construction rules. In this paper, we extend the method in \cite{lindemann2018control} to HOCBFs and provide a general and algorithmic procedure to construct these HOCBFs. Further, we make these HOCBFs trainable so that the manual design is avoided, and the performance of the controller including these HOCBFs can be further improved through training.  

Consider an STL formula $\varphi$ as in \eqref{eq:stl}. Since for all predicates with negations $\neg\mu$ we can replace the predicate function with $-h(\mathbf x)$ and remove the negation, we assume that the formula $\varphi$ is negation-free without loss of generality. We make the following assumption on the STL formula and the system:
\begin{assumption}
\label{as:feasible}
$\forall \mathbf x(0)\in\mathcal X_0$, $\exists \mathbf u(t)\in\mathcal U$ such that $(\mathbf x, 0)\models\varphi$ where $\mathbf x$ is the solution of system \eqref{eq:system}.
\end{assumption}
Assumption \ref{as:feasible} is not restrictive in practice since if it is not true, for some $\mathbf x_0$ there is no solution for Problem \ref{pb:1}. 

\noindent\textbf{Categories of Predicates.} Suppose that there are $M$ predicates in $\varphi$ and they are given by $\mu_j:\ h_j(\mathbf x)\geq0$, $j=1,\ldots,M$. Now we divide all predicates into three categories:
\begin{itemize}
    \item Category I: predicates that are satisfied at $t=0$ and the starting time of the temporal operator wrapping it is $0$, e.g., $\mu_1$ in $G_{[0,5]}\mu_1$ and $G_{[0,5]}\mu_1\land\mu_2$, where $h_1(\mathbf x_0)\geq0$. These predicates usually define safety requirements, such as obstacle avoidance in robotic applications.
    \item Category II: All predicates wrapped by $F_{[t_a,t_b]}$ that do not belong to Category I, e.g., $\mu_1$ in $F_{[2,5]}\mu_1$ and $\mu_2$ in $F_{[0,5]}\mu_2\land\mu_3$ where $h_2(\mathbf x_0)<0$.
    \item Category III: All predicates wrapped by $G_{[t_a,t_b]}$ that do not belong to Category I,  e.g., $\mu_1$ in $G_{[2,5]}\mu_1\land\mu_2$. Note that Assumption \ref{as:feasible} avoids formulae like $G_{[0,5]}\mu_1$, where $h_1(\mathbf x_0)<0$. 
\end{itemize}

\noindent\textbf{STL Guarantees.} To each predicate $\mu_j$, we assign a (time-varying) HOCBF $b_j$. Since each predicate $\mu_j$ belonging to Category I has already been satisfied at $t=0$, we assign a fixed and time-invariant HOCBF to retain its satisfaction for the required time: 
\begin{equation}
\label{eq:cat1}
b_j(\mathbf x) = h_j(\mathbf x).
\end{equation}
For predicates $\mu_j$ in Category II and III, we assign a trainable time-varying HOCBF:
\begin{equation}
    \label{eq:cat23}
    b_j(\mathbf x,t,\bm\theta_b,\mathbf x_0) = h_j(\mathbf x) + \gamma_j (t,\bm\omega_j(\bm\theta_b,\mathbf x_0)),
\end{equation}
where $\gamma_j(\cdot,\bm\omega_j):[0,T]\rightarrow\mathbb R$ is a function parameterized by $\bm\omega_j$, $\bm\omega_j$ is given by a neural network with input $\mathbf x_0$ and parameters $\bm\theta_b$. Details about this neural network will be discussed in Section \ref{subsec:learning}. In the rest of this subsection, we will omit $\bm\theta_b$ and $\mathbf x_0$ for notation simplicity and just consider $\bm\omega_j$ as a vector. By properly choosing $\gamma_j(t,\bm\omega_j)$, the satisfaction of $b_j(\mathbf x,t)\geq0$, $\forall t\in[0,T]$ can ensure the satisfaction of the predicate $\mu_j$ during the required time slots. Next, we discuss the selection of $\gamma_j(t,\bm\omega_j)$.

For simplicity of notation, we omit the subscript $j$ when it is clear from the context. For a predicate $\mu$ in Category II that is wrapped with $F_{[t_a,t_b]}$, we choose $\gamma$ to be a linear function:
\begin{equation}
    \label{eq:gamma_f}
    \gamma(t,\bm\omega) = \omega_1 + \omega_2t,
\end{equation}
where $\bm\omega = (\omega_1,\omega_2)$, $\omega_1>0$, $\omega_2<0$. Note that other forms of functions are also possible. To make sure the HOCBF $b(\mathbf x,t) = h(\mathbf x) + \gamma(t,\bm\omega)$ guarantees the satisfaction of $\mu$, we add $3$ constraints on $\gamma$:
\begin{subequations}
\label{eq:f}
    \begin{align}
    \label{eq:f0}
        &\gamma(0,\bm\omega)\ > -h(\mathbf x_0),\\
    \label{eq:fb}
        &\gamma(t_b, \bm\omega)  \leq 0,\\
    \label{eq:fa}
        &\gamma(t_a, \bm\omega) > -\sup_{\mathbf x\in\mathbb R^n} h(\mathbf x).
    \end{align}
\end{subequations}
Constraint \eqref{eq:f0} ensures the HOCBF is positive at the initial time, i.e., $b(\mathbf x_0,0)>0$. Constraint \eqref{eq:fb} ensures that $h(\mathbf x) \geq b(\mathbf x,t)\geq0$ before time $t_b$. Given \eqref{eq:f0} and \eqref{eq:fb} the forward invariance of the superlevel set of $b(\mathbf x,t)$ enforces the satisfaction of $F_{[t_a,t_b]}\mu$. The third constraint \eqref{eq:fa} ensures that the superlevel set of $b(\mathbf x, t)$ is nonempty when $t<t_a$. As it will be dicussed later, we delete the HOCBF once $h(\mathbf x)>0$ when $t\geq t_a$, so we do not consider whether the superlevel set of $b(\mathbf x, t)$ is empty after $t_a$.

For a predicate $\mu$ in Category III that is wrapped with $G_{[t_a,t_b]}$, let $\gamma$ be defined as:
\begin{equation}
    \label{eq:gamma_g}
    \gamma(t,\bm\omega) = \omega_1 e^{-\omega_2t}-c,
\end{equation}
where $\bm\omega = (\omega_1,\omega_2)$, $\omega_1>0$, $\omega_2>0$. $c>0$ is a small constant. Again, other forms of functions are possible. Similar to \eqref{eq:f}, we have two constraints on $\gamma$:
\begin{subequations}
\label{eq:g}
    \begin{align}
    \label{eq:g0}
        &\gamma(0,\bm\omega)\ > -h(\mathbf x_0),\\
    \label{eq:ga}
        &\gamma(t_a,\bm\omega)  \leq 0.
    \end{align}
\end{subequations}
The difference is that \eqref{eq:ga} ensures $h(\mathbf x) \geq b(\mathbf x,t)
\geq0$ before time $t_a$ so that $G_{[t_a,t_b]}\mu$ is enforced to be satisfied. When $c>0$ is small enough, the superlevel set of $b(\mathbf x,t)$ is always nonempty under Assumption \ref{as:feasible}. We choose the exponential function \eqref{eq:gamma_g} for \emph{always} instead of a linear function because it satisfies:
\begin{equation*}
    0\leq-\gamma(t,\bm\omega)<c,\ \forall t\in[t_a,t_b].
\end{equation*}
As a result, $b(\mathbf x,t)\geq 0$, i.e., $h(\mathbf x)\geq-\gamma(t,\bm\omega)$, is not over-conservative for $t\in[t_a,t_b]$ when $c>0$ is small enough. As it will be detailed below, the HOCBF is deleted when $t>t_b$, which further mitigates over-conservativeness.

\noindent\textbf{Addressing Conflicts between HOCBFs.} We can construct an HOCBF $b_j$ for each predicate $\mu_j$ in $\varphi$ using \eqref{eq:cat1} or \eqref{eq:cat23}. However, it is possible that the corresponding constraints $b(\mathbf x,t)\geq0$ are conflicting with each other during some time periods. Here, we propose a solution to this problem. We first make an additional assumption:
\begin{assumption}
\label{as:circle}
    Let all predicate functions in Category II and III be in the form of:
    \begin{equation}
    \label{eq:predicate}
        h(\mathbf x) = \pm\big(R-\|l(\mathbf x) - \mathbf o\|_2\big),
    \end{equation}
    where $l:\mathbb R^n\rightarrow \mathbb R^o$ is a differentiable function shared by all predicates mapping state $\mathbf x$ to a vector that we care about, e.g., the location of a robot, and $R\in\mathbb R_{+}$, $\mathbf o\in\mathbb R^o$ are the radius and center of a circular region. We denote the interior (including boundary) and exterior of this region as $\mathcal B(\mathbf o,R)$ and $\mathcal B^{\complement}(\mathbf o,R)$ respectively, where the superscript $\complement$ indicates the complement set in $\mathbb R^o$. 
\end{assumption}

Predicates in the form of \eqref{eq:predicate} can express tasks of reaching ($+$) or avoiding ($-$) a circular region. Together with temporal operators, they can specify rich requirements. We will explain why we assume circular regions. Other type of predicates in Category II and III will be investigated in future work. Next, we give an example to illustrate the idea. 

\begin{example}
Consider a formula $F_{[0,2]}\mu_1\land F_{[2,4]}\mu_2$, where  $h_1(\mathbf x) = R_1-\|l(\mathbf x) - \mathbf o_1\|_2$ and $h_2(\mathbf x) = R_2-\|l(\mathbf x) - \mathbf o_2\|_2$. The invariant sets of the corresponding HOCBFs are $\mathcal B(R_1+\gamma_1(t),\mathbf o_1)$ and $\mathcal B(R_2+\gamma_2(t),\mathbf o_2)$ at time $t$ respectively. To avoid conflicts, (1) we require that these two regions have an nonempty intersection for all $t\in[0,2]$, and (2) we delete $b_1$ once $h_1(\mathbf x)>0$. For the former, it is sufficient to require $\mathcal B(R_1+\gamma_1(2),\mathbf o_1)\cap \mathcal B(R_2+\gamma_2(2),\mathbf o_2)\neq \emptyset$, that is, $\gamma_2(2)\geq\|\mathbf o_1,\mathbf o_2\|_2 - \gamma_1(2) - R_1 - R_2$.
\end{example}

Specifically, for each predicate $\mu$ wrapped with $G_{[t_a,t_b]}$, we delete the corresponding HOCBF at time $t=t_b$. For predicates in the form of $F_{[t_a,t_b]}\mu$, we delete the corresponding HOCBFs once $h(\mathbf x)>0$ after $t_a$. For predicates in the form of $F_{[t_a,t_b]}\land_{j=1}^N\mu_j$, we delete the corresponding HOCBFs for all $\mu_j$ together once $h_j(\mathbf x)>0$ for all $j$ after $t_a$. Note that $G_{t_a,t_b}(\mu_1\land\mu_2)$ is equivalent to $G_{[t_a,t_b]}\mu_1\land G_{[t_a,t_b]}\mu_2$, but $F_{t_a,t_b}(\mu_1\land\mu_2)$ is different from $F_{[t_a,t_b]}\mu_1\land F_{[t_a,t_b]}\mu_2$. The latter allows asynchronous satisfaction of $\mu_1$ and $\mu_2$.

Next, we reorder all predicates $\mu_j$ according to the ending time points $t_b^j$ of the temporal operators wrapping them such that $t_b^1\leq t_b^2\leq \ldots\leq t_b^M$. For the predicate $\mu_j$ with ending time $t_b^j$, $j=2,\ldots,M$, we add $j-1$ additional constraints besides \eqref{eq:f} or \eqref{eq:g}:
\begin{equation}
\label{eq:add_cons}
\begin{aligned}
    &\gamma_j(t_b^1) \geq s_1s_j\|\mathbf o_1,\mathbf o_j\|_2 - \gamma_1(t_b^1) - s_1R_1 - s_jR_j + D_{j,1},\\ 
    &\gamma_j(t_b^2) \geq s_2s_j\|\mathbf o_2,\mathbf o_j\|_2 - \gamma_2(t_b^2) - s_2R_2 - s_jR_j+D_{j,2},\\ 
    & \qquad\qquad\qquad \vdots\\
    &\gamma_j(t_b^{j-1}) \geq s_{j-1}s_j\|\mathbf o_{j-1},\mathbf o_j\|_2 - \gamma_{j-1}(t_b^{j-1}) \\
    &\qquad\qquad\qquad- s_{j-1}R_{j-1} - s_jR_{j} + D_{j,j-1},
\end{aligned}
\end{equation}
where $\bm\omega$ is omitted, $h_k(\mathbf x) = s_k\big(R_k-\|l(\mathbf x) - \mathbf o_k\|_2\big)$, $s_k\in\{-1,1\}$, $D_{j,k}=min(s_k+s_j,0)\times \inf$, i.e., when $s_k=s_j=-1$, we release the constraint. For predicates in Category I, let $\gamma(t)=0$ for all $t\in[0,T]$. Intuitively, consider $\mu_1$ and $\mu_2$ with $t_b^1\leq t_b^2$. If both of them are reachability requirements, \eqref{eq:add_cons} enforces $\mathcal B(R_1+\gamma_1(t_b^1),\mathbf o_1)\cap \mathcal B(R_2+\gamma_2(t_b^1),\mathbf o_2)\neq\emptyset$. If $\mu_1$ is a reachability task and $\mu_2$ is an avoidance task, \eqref{eq:add_cons} enforces $\mathcal B(R_1+\gamma_1(t_b^1),\mathbf o_1)\cap \mathcal B^\complement(R_2-\gamma_2(t_b^1),\mathbf o_2)\neq\emptyset$. If $\mu_1$ is an avoidance task and $\mu_2$ is a reachability task, \eqref{eq:add_cons} enforces $\mathcal B^\complement(R_1-\gamma_1(t_b^1),\mathbf o_1)\cap \mathcal B(R_2+\gamma_2(t_b^1),\mathbf o_2)\neq\emptyset$. If both are avoidance tasks, no extra condition is needed as $\mathcal B^\complement(R_1-\gamma_1(t_b^1),\mathbf o_1)\cap \mathcal B^\complement(R_2-\gamma_2(t_b^1),\mathbf o_2)$ is always nonempty.

With the construction of the HOCBFs and the corresponding constraints described above, we have:

\begin{theorem}
\label{thm:construct}
    Assume we have a STL formula $\varphi$, a system \eqref{eq:system} satisfying Assumptions \ref{as:feasible} and \ref{as:circle}, a set of HOCBFs constructed by \eqref{eq:cat1} and \eqref{eq:cat23} that satisfy all constraints \eqref{eq:f}, \eqref{eq:g} and \eqref{eq:add_cons}, and a sequence of functions $\psi_i$ for each HOCBF as in \eqref{eq:psi}, where $\psi_i(\mathbf x_0,0)\geq0$, $i=1,\ldots,m$. Then a control law $\mathbf u(t)$ that satisfies \eqref{eq:hocbf} for all HOCBFs is guaranteed to satisfy specification $\varphi$. 
\end{theorem}

\begin{proof}
Constraints \eqref{eq:f0} and \eqref{eq:g0} ensure that $b(\mathbf x_0,0)\geq 0$ for all HOCBFs. Since $\psi_i(\mathbf x_0,0)>0$, $i=1,\ldots,m$ for all HOCBFs, according to Theorem \ref{thm:cbf}, a control law $\mathbf u(t)$ that satisfies \eqref{eq:hocbf} for all $b_j(\mathbf x,t)$ ensures $b_j(\mathbf x,t)\geq0$, for all $b_j$ that have not been deleted at time $t$. Since $\gamma(t)$ in \eqref{eq:gamma_f} and \eqref{eq:gamma_g} are non-decreasing, \eqref{eq:fb} ensures $\exists t'\in[t_a,t_b]$, $h(\mathbf x(t'))\geq0$ for $F{[t_a,t_b]}\mu$, while \eqref{eq:ga} ensures $\forall t'\in[t_a,t_b]$, $h(\mathbf x(t'))\geq0$ for $G{[t_a,t_b]}\mu$. For a formula in the form of $F_{[t_a,t_b]}\land_{j=1}^N\mu_j$, \eqref{eq:fb} ensures $\exists t'\in[t_a,t_b]$, $h_j(\mathbf x(t'))\geq0$ for all $j$. Hence, all predicates are satisfied at the required time, and the STL specification $\varphi$ is satisfied.
\end{proof}

\subsection{Learning Robust Controllers}
\label{subsec:learning}
Theorem \ref{thm:construct} ensures the satisfaction of the STL specification when all HOCBFs constraints are satisfied. Then we can use BarrierNet \eqref{eq:barriernet} to obtain a controller that satisfies all HOCBFs constraints. In this subsection, we first explain why in \eqref{eq:barriernet}, $b(\mathbf x,t,\bm\theta_b,\mathbf x_0)$ and $p_{i,j}(\mathbf x_0,\bm\theta_p)$ all depend on the initial condition $\mathbf x_0$. Then we describe the structure of the entire neural network controller $\pi(\mathbf x_{0:t},\bm\theta)$. Finally, we introduce the training process of the controller. 

\begin{figure}
    \centering
    \includegraphics[width=8cm]{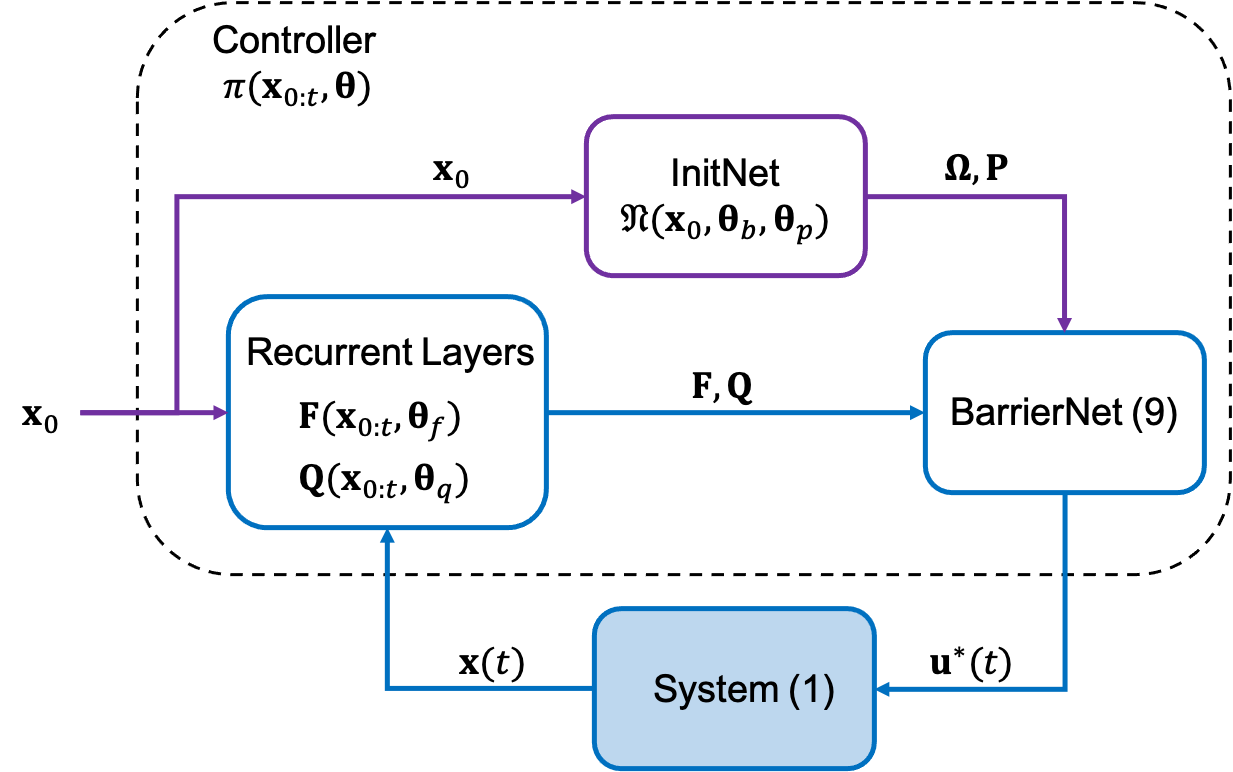}
    \caption{Overall structure of the controller. Purple parts are only executed at $t=0$, while blue parts are executed repeatedly. The dashed box indicates the controller $\pi(\mathbf x_{0:t},\bm\theta)$.}
    \label{fig:overall}
\end{figure}

\noindent\textbf{Parameters Depending on Initial Conditions.} Consider a predicate belonging to Category II or III with corresponding HOCBF $b(\mathbf x,t) = h(\mathbf x) + \gamma(t,\bm\omega)$. Since constraints \eqref{eq:f0} and \eqref{eq:g0} on parameters $\bm\omega$ contain the initial condition $\mathbf x_0$, different $\bm\omega$ should be used for different initial condition $\mathbf x_0$. Hence, we use a neural network whose input is $\mathbf x_0$ to provide $\bm\omega$, denoted as $\bm\omega(\mathbf x_0,\bm\theta_b)$. As a result, the HOCBF also depends on $\mathbf x_0$ and contains trainable parameters $\bm\theta_b$, denoted as $b(\mathbf x,t,\bm\theta_b,\mathbf x_0)$.

On the other hand, to use HOCBFs to guarantee set-invariance, we also need to make sure $\psi_i(\mathbf x_0,0)\geq0$ for all $i=1,\ldots,m$. Since $b(\mathbf x_0,0)> 0$, we can always find a large enough $p_i$ such that $\psi_i(\mathbf x_0,0)\geq0$ according to \eqref{eq:psi'}. These constraints on $p_i$ also depend on $\mathbf x_0$. Hence, we use a neural network with input $\mathbf x_0$ and parameters $\bm\theta_p$ to provide $p_i$, denoted as $p_i(\mathbf x_0,\bm\theta_p)$.

\noindent\textbf{Neural Network Controller Structure.} In practice, we use one neural network referred to as InitNet to provide all parameters depending on $\mathbf x_0$:
\begin{equation}
    \label{eq:nn_init}
    [\bm\Omega^\top \mathbf P^\top] = \mathfrak N(\mathbf x_0, \bm\theta_b,\bm\theta_p),
\end{equation}
where $\bm\Omega=[\bm\omega_1^\top\ldots\bm\omega_N^\top]^\top\in\mathbb R^{2N}$, $\mathbf P$ is the concatenation of all $p_i$ in \eqref{eq:psi'} for all HOCBFs, $\mathfrak N$ is the neural network parameterized by trainable parameters $\bm\theta_b$ and $\bm\theta_p$. We transform constraints on $\gamma$ \eqref{eq:f}, \eqref{eq:g} and \eqref{eq:add_cons} into constraints on $\bm\Omega$. 
For constraints in the form of $\omega\in[\underline\omega,\overline\omega]$ we apply a Sigmoid function on the last layer of $\mathfrak N$ while for constraints in the form of $\omega\in[\underline\omega,\infty)$ or $\omega\in(-\infty,\overline\omega]$ we apply a Softplus function. In this way, $\bm\Omega$ satisfies all constraints in \eqref{eq:f}, \eqref{eq:g} and \eqref{eq:add_cons}. Similarly, for $p_i$, $i=1,\ldots, m-1$, we add constraints $p_i > max\big[-\dot\psi_{i-1}(\mathbf x,0)/\alpha_i\big(\psi_{i-1}(\mathbf x,0)\big),0\big]$ which are also implemented by Softplus functions. 

InitNet is only used at time $t=0$ to provide a set of HOCBFs and the corresponding class $\mathcal K$ functions, which are fixed after $t=0$. Then we use another (recurrent) neural network parameterized by $\bm\theta_q$ and $\bm\theta_f$ to provide $\mathbf Q(\mathbf x_{0:t},\bm\theta_q)$ and $\mathbf F(\mathbf x_{0:t},\bm\theta_f)$ at each discrete time point. The whole controller $\pi(\mathbf x_{0:t},\bm\theta)=\mathbf u^*$  contains $\mathbf Q(\mathbf x_{0:t},\bm\theta_q)$, $\mathbf{F}(\mathbf x_{0:t},\bm\theta_f)$, $\mathfrak N(\mathbf x_0, \bm\theta_b,\bm\theta_p)$ and the dQP \eqref{eq:barriernet} with $\bm\theta=(\bm\theta_q,\bm\theta_f,\bm\theta_b,\bm\theta_p)$. The overall structure of the controller is shown in Fig. \ref{fig:overall}

\noindent\textbf{Training BarrierNet.} Similar to \cite{liu2021safe}, we randomly sample $V$ initial conditions $\mathbf x_0^v$, $v=1,\ldots,V$. We apply the system dynamics \eqref{eq:system} with the controller $\pi$ until reaching the time horizon $T$ to get $V$ state and control trajectories. We evaluate their STL robustness and cost $J$, and then use the mean value to approximate the expectation. Formally, we rewrite \eqref{eq:goal} into:
\begin{equation}
    \label{eq:cost}
    \begin{aligned}
       \bm\theta^* &= \arg\max_{\bm\theta} \frac{1}{V}\sum_{v=1}^V\big[\rho(\varphi,\mathbf x^v, 0) - J(\mathbf u^v)\big] \\
       \text{s.t.}\ & \dot{\mathbf x}^v = f(\mathbf x^v) + g(\mathbf x^v)\pi(\mathbf x^v_{0:t},\bm\theta),\ v=1,\ldots,V,
    \end{aligned}
\end{equation}
where the superscript $v$ indicates the $v^{th}$ sample. We substitute the constraint (dynamics) into the objective function to make it an unconstrained optimization problem. Since the QP \eqref{eq:barriernet} is differentiable with respect to its parameters using the technique in \cite{amos2017optnet}, we backpropagate the gradient of the objective funtion in \eqref{eq:cost} through the QP to all parameters $\bm\theta$. The gradients of the STL robustness are calculated analytically and automatically using an adapted version of STLCG \cite{leung2020backpropagation} that use the robustness in \cite{liu2022robust}. Then we update the parameters using the gradient. Note that at each optimization step we randomly resample $V$ initial conditions to have a better exploration of the initial set $\mathcal X_0$ and we use the stochastic optimizer Adam \cite{kingma2014adam} to train the parameters.

The following corollary from Theorem \ref{thm:construct} states that our network controller is correct: 

\begin{corollary}
    \label{thm:correct}
    Consider an STL formula $\varphi$ and a system \eqref{eq:system} satisfying Assumptions \ref{as:feasible} and \ref{as:circle}. Then any neural network controller with BarrierNet \eqref{eq:barriernet} as the last layer guarantees that the solution of system \eqref{eq:system} starting from any initial condition $\mathbf x_0\in\mathcal X_0$ satisfies the specification $\varphi$. 
\end{corollary}
\begin{proof}
    This follows immediately from Theorem \ref{thm:construct}.
\end{proof}
\begin{remark}
    Since we can only feed discrete data into neural networks, the QP \eqref{eq:barriernet} is solved in discrete time. Hence, the inter-sampling effect should also be considered to guarantee the correctness of the controller. A possible solution to address this effect is the event-triggered framework \cite{xiao2021event}. We will consider this in future work. 
\end{remark}

We summarize our solution to Problem \ref{pb:1} in Algorithm \ref{alg:1}.
\vspace{-15pt}
\begin{algorithm}
\KwIn{System dynamics \eqref{eq:system} and STL formula $\varphi$}
\KwOut{Robust and correct controller $\pi(\mathbf x_{0:t},\bm\theta^*)$}
Construct HOCBFs from $\varphi$ using \eqref{eq:cat1}, \eqref{eq:cat23}\;
Set up constraints on $\bm\omega$ using \eqref{eq:f}, \eqref{eq:g}, \eqref{eq:add_cons}\;
Set up constraints on $p_{i,j}$ using \eqref{eq:psi'}\;
Initialize controller $\pi(\mathbf x_{0:t},\bm\theta)$ including $\mathbf Q(\mathbf x_{0:t},\bm\theta_q)$, $\mathbf{F}(\mathbf x_{0:t},\bm\theta_f)$, InitNet \eqref{eq:nn_init} and the dQP \eqref{eq:barriernet}\;
\Repeat{Convergence; \Return $\bm\theta^*$}{
Sample $V$ initial conditions $x_0^v$\;
Obtain $\bm\Omega$, $\mathbf P$ for each $x_0^v$ from InitNet \eqref{eq:nn_init}\;
Evaluate \eqref{eq:cost} by applying $\pi(\mathbf x_{0:t},\bm\theta)$ to \eqref{eq:system}\;
Compute gradient of \eqref{eq:cost} w.r.t. $\bm\theta$\;
Update $\bm\theta$ using Adam optimizer\;
}
 \caption{Construction and training of controller}\label{alg:1}
\end{algorithm}
\vspace{-10pt}

\section{Simulations}
\label{sec:sim}
In this section, we demonstrate the efficacy of our approach via simulations and compare it with existing algorithms. 

\noindent\textbf{Environment and STL setup.} Consider a 2D robot navigation problem. The dynamics of the robot is given as:
\begin{equation}
    \label{eq:robot}
    \begin{bmatrix} \dot{p}_x\\ \dot{p}_y\\ \dot{v}_x\\ \dot{v}_y\end{bmatrix} = \begin{bmatrix} v_x\\ v_y\\ 0\\ 0\end{bmatrix} + \begin{bmatrix} 0&0\\ 0&0\\ 1&0\\ 0&1\end{bmatrix} \begin{bmatrix}a_x\\a_y\end{bmatrix},
\end{equation}
where $\mathbf x=[p_x\ p_y\ v_x\ v_y]^\top$, $\mathbf u=[a_x\ a_y]^\top$, $[p_x\ p_y]^\top$ is the 2D position, $[v_x\ v_y]^\top$ is the velocity, and $[a_x\ a_y]^\top$ is the acceleration of the robot. We assume the control has no bounds in this case and use the L2 norm for the cost function in \eqref{eq:cost} with a coefficient of $0.003$ to punish large accelerations. Consider the environment shown in Fig. \ref{fig:env}. $\mathbf x_0$ is uniformly sampled in the region $Init$ with zero velocity. We discretize the system with a time interval of $0.1s$. The task for the robot is given by an STL formula:
\begin{equation}
    \label{eq:task}
    \varphi = F_{[0,2]}Reg_1 \land F_{[2,5]}Reg_2 \land G_{[0,5]}(\neg Obs_1 \land \neg Obs_2),
\end{equation}
where $Reg_i$ indicates $R_i-\|l(\mathbf x) - \mathbf o_i\|_2\geq 0$, $i=1,2$, $l(\mathbf x) = [p_x\ p_y]^\top$. $Obs_i$ is a superellipse:
\begin{equation}
    \label{eq:obs}
    1-\sqrt[4]{(\frac{p_x-o_{x,i}}{a_i})^4 + (\frac{p_y-o_{y,i}}{b_i})^4}\geq 0,
\end{equation}
$i=1,2$. Here, $Reg_i$ belongs to Category II and $Obs_i$ belongs to Category I, $i=1,2$. In plain English, the STL formula $\varphi$ requires the robot to eventually visit $Reg_1$ within $[0,2]$ and eventually visit $Reg_2$ within $[2,5]$, while always avoid obstacles $Obs_1$ and $Obs_2$. The time horizon of $\varphi$ is $5$. For all $4$ predicates, the corresponding HOCBFs have a relative degree of $2$ with respect to system \eqref{eq:robot}. In this example, we fixed $\mathbf Q(\mathbf x, \bm\theta_q)$ to an identical matrix, so the output of the previous layers at $t>0$ is just $\mathbf F(\mathbf x, \bm\theta_f)$, which can be interpreted as a reference control. Since this task does not require back and forth motions, a recurrent neural network is not necessary for $\mathbf F$. Hence, both $\mathbf F$ and InitNet are implemented as neural networks with $3$ fully connected layers. For the robustness function, we use the exponential robustness given in \cite{liu2022robust}. Specifically, for a predicate $\mu:h(\mathbf x)\geq0$, let $\rho(\mu, \mathbf x, t) = h(\mathbf x(t))$, $\rho(\neg\mu, \mathbf x, t) = -h(\mathbf x(t))$. Since ``always" and ``eventually" can be regard as conjunction and disjunction over time, and disjunction can be replaced by conjunction and negation using De Morgan law, we only need to give the definition of the robustness for conjunctions. Consider the conjunction over $M$ subformulas with robustness $\rho_1,\ldots,\rho_M$. We first define an effective robustness measure, denoted by $\rho_i^{conj}$, $i=1,\ldots,M$, for each subformula:
\begin{equation}
\label{eq:exp-conj-eff}
    \rho_{i}^{conj} \coloneqq \left\{
    \begin{aligned}
    &\rho_{min}e^{\frac{\rho_i-\rho_{min}}{\rho_{min}}} \quad & \rho_{min} < 0 \\
    &\rho_{min}(2 - e^{\frac{\rho_{min}-\rho_i}{\rho_{min}}}) \quad & \rho_{min} > 0 \\
    & 0 \quad & \rho_{min} = 0
    \end{aligned}\right.
\end{equation}
where $\rho_{min}=\min(\rho_1,\ldots,\rho_M)$. 
Then the exponential robustness for conjunction is defined as:
\begin{equation}
\label{eq:exp-and}
    \mathcal A^{exp}(\rho_1,\ldots,\rho_M) = \beta\rho_{min} +  (1-\beta)\frac{1}{M}\sum_{i=1}^M \rho_i^{conj},
\end{equation}
where $\beta\in[0,1]$ balances the contribution between $\rho_{min}$ and the mean of $\rho_i^{conj}$ (same sign as $\rho_{min}$). The exponential robustness is sound in the sense that $\rho(\varphi,\mathbf x,0)\geq0$ if and only if $(\mathbf x,0)\models \varphi$. More details about the robustness function can be found in \cite{liu2022robust}. 

\begin{figure}
     \centering
     \begin{subfigure}[b]{0.235\textwidth}
         \centering
         \includegraphics[height=4.8cm]{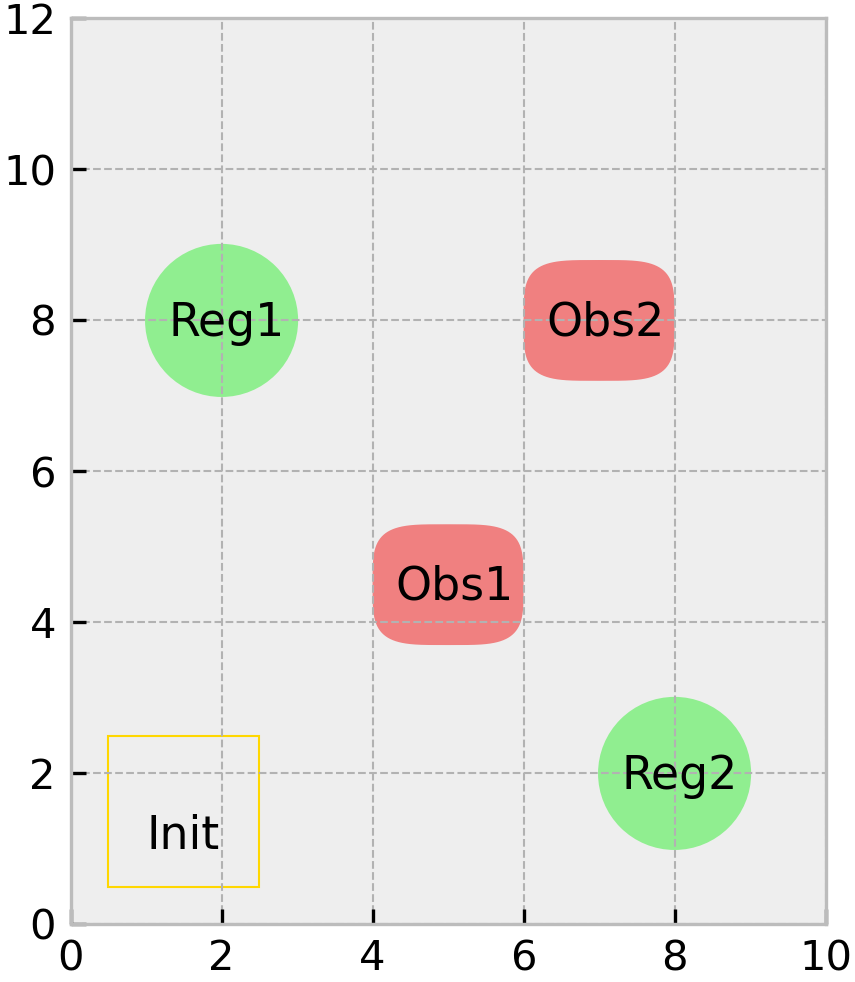}
         \caption{\small Environment}
         \label{fig:env}
     \end{subfigure}
     \begin{subfigure}[b]{0.235\textwidth}
         \centering
         \includegraphics[height=4.8cm]{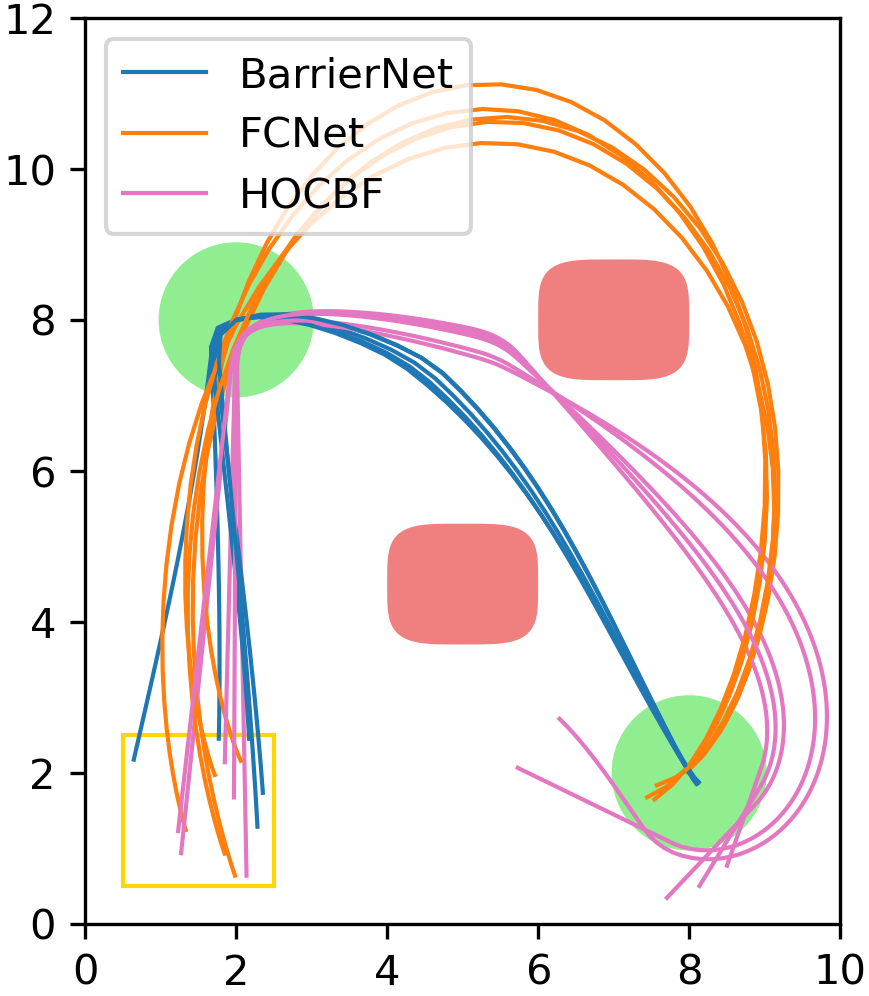}
         \caption{\small Trajectories}
         \label{fig:traj}
     \end{subfigure}
    \caption{(a) The 2D environment we consider. (b) Sampled trajectories with random initial conditions using three methods: BarrierNet (developed in this paper), FCNet, and HOCBFs. }
\label{fig:trajorties}
\end{figure}

\noindent\textbf{Comparison setup.} We construct the HOCBFs and train the controller proposed in this paper. 
We compare the results with our previous work \cite{liu2021safe} where a neural network controller without BarrierNet, i.e., a Fully Connected Neural Network (FCNet) is trained for an STL task. It is equivalent to directly use the reference control $\mathbf F$. We refer to these two controllers as BarrierNet and FCNet respectively. To make the comparison fair, we assume that the system dynamics are known for \cite{liu2021safe}. We use the same objective function, optimizer, and the same neural network architectures, i.e., the FCNet has the same structure with $\mathbf F(\mathbf x, \bm\theta_f)$. Training curves are illustrated in Fig. \ref{fig:curves}. Meanwhile, we directly apply the approach in \cite{lindemann2018control} (extended to HOCBF) without any learning, i.e., we construct a set of HOCBFs with fixed parameters and solve the QP \eqref{eq:barriernet} with $\mathbf F=\mathbf 0$. The parameters are randomly chosen but satisfy all constraints \eqref{eq:f}, \eqref{eq:g} and \eqref{eq:add_cons}. We refer to this approach as HOCBF. The resulting average values of the objective function and the robustness starting from random initial conditions are shown in Fig.\ref{fig:curves} with dashed lines. Sampled trajectories using the three approaches with random initial conditions are shown in Fig. \ref{fig:traj}.

\noindent\textbf{Analysis and Discussion.} In Fig. \ref{fig:ro} we can see that when using BarrierNet, the robustness value is positive (which means the STL specification is satisfied) from the beginning of the training. This demonstrates the correctness of Corollary \ref{thm:correct}. As for FCNet, it takes about $150$ iterations to get a positive mean robustness value. The results of directly applying HOCBFs with randomly chosen parameters are similar as using BarrierNet with an untrained neural network, i.e., at the first iteration during training. It also satisfies the specification but is less robust than using BarrierNet after training. As shown in Fig. \ref{fig:traj} the robot reaches the center of $Reg_2$ with both BarrierNet and FCNet after training while only reaches the boundary of $Reg_2$ when directly using HOCBFs. The robot leaves $Reg_2$ after the corresponding HOCBF is deleted. The final robustness and objective function values of BarrierNet are both higher than FCNet. Without the guidance of HOCBFs, the FCNet controller only finds a sub-optimal solution which steers the robot further away to avoid the obstacles. Since QP can be solve very efficiently, all three methods can execute fast during testing which indicates the ability of real-time control. 

\begin{figure}
     \centering
     \begin{subfigure}[b]{0.235\textwidth}
         \centering
         \includegraphics[height=4.8cm]{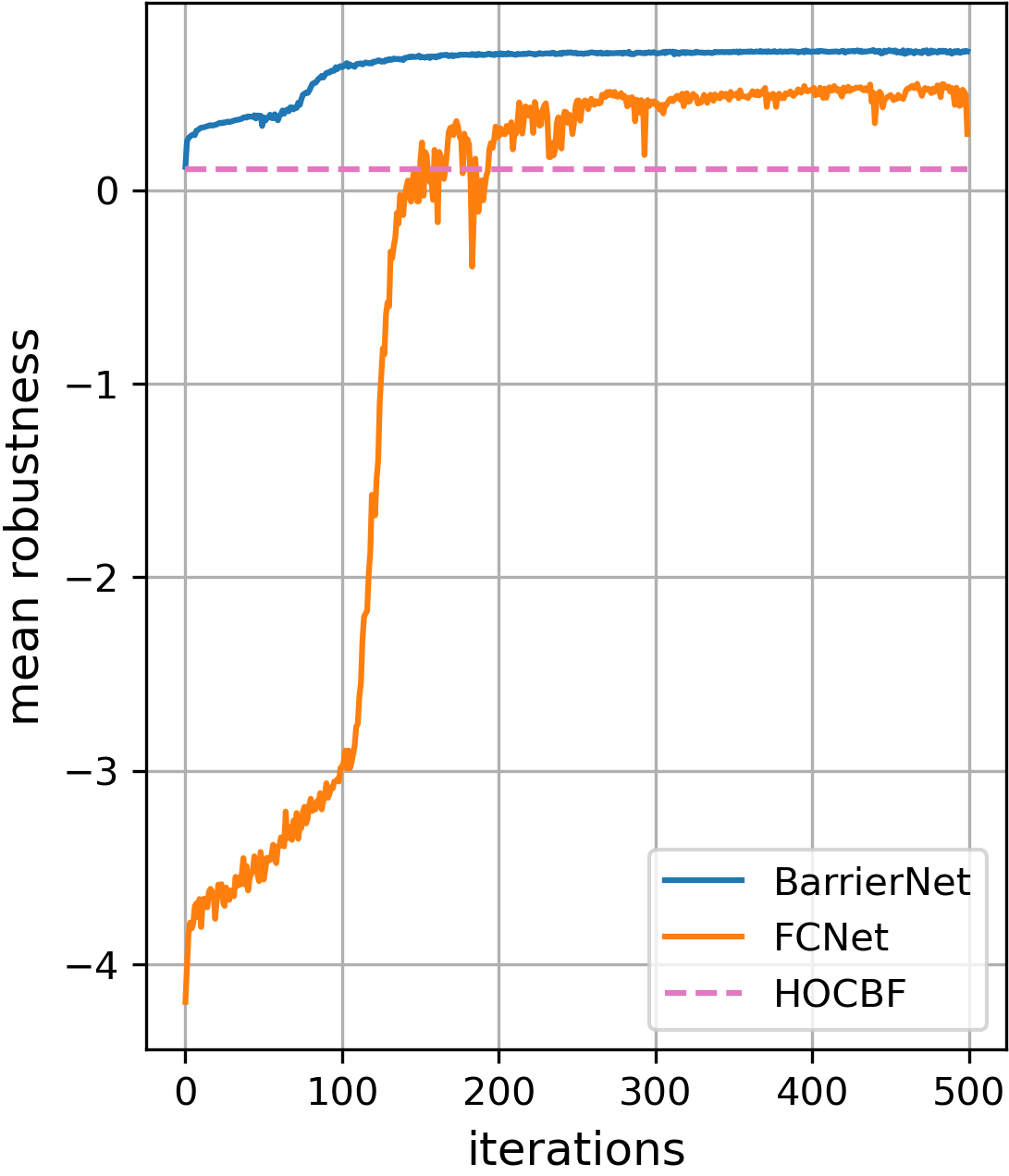}
         \caption{\small Robustness}
         \label{fig:ro}
     \end{subfigure}
     \ 
     \begin{subfigure}[b]{0.235\textwidth}
         \centering
         \includegraphics[height=4.8cm]{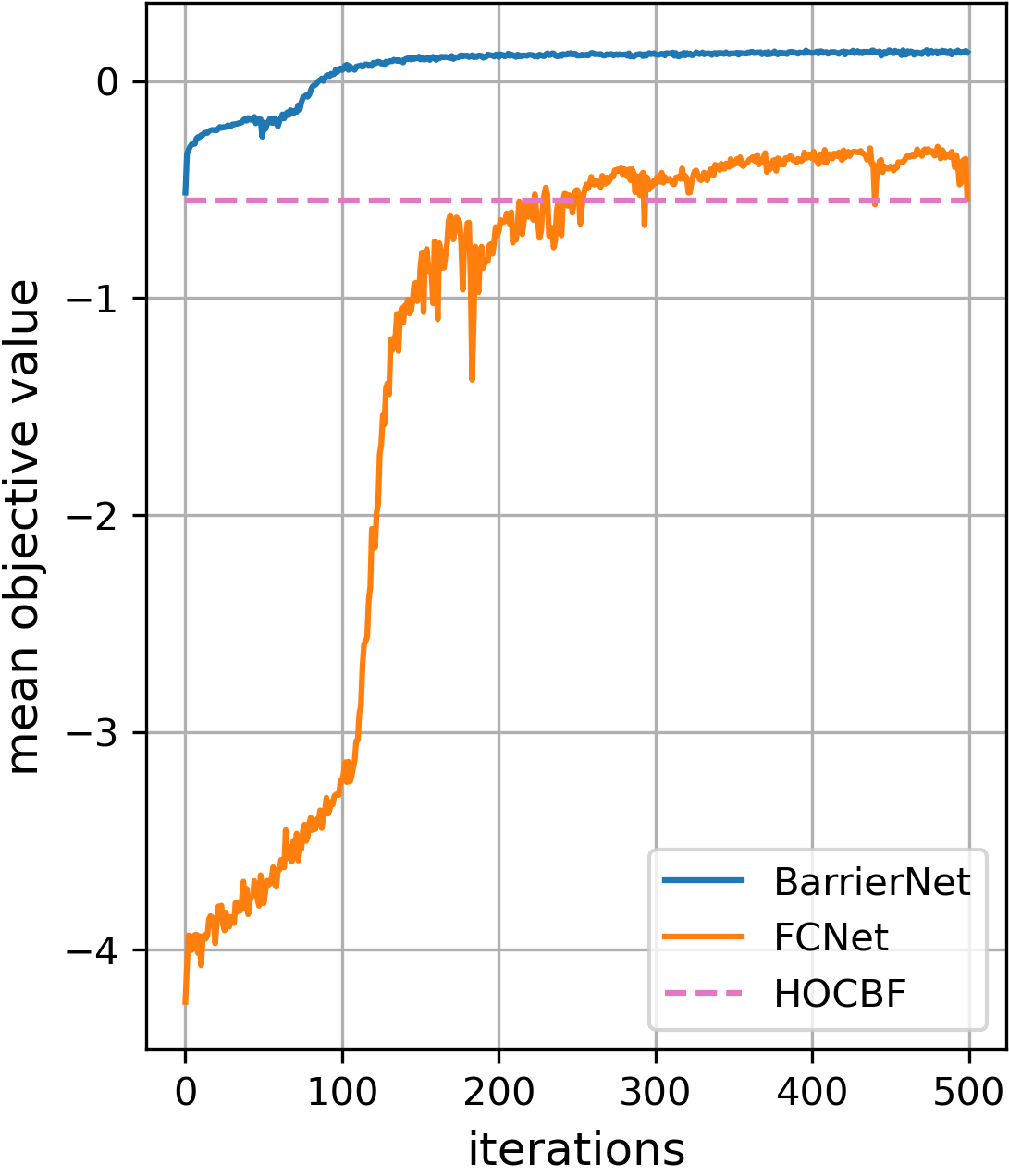}
         \caption{\small Objective}
         \label{fig:ob}
     \end{subfigure}
    \caption{Learning curves for the methods of BarrierNet and FCNet. Dashed lines show the result of directly using HOCBFs. (a) the mean robustness values during training. (b) the mean objective function values during training.}
\label{fig:curves}
\end{figure}

\section{Conclusion and Future Work}
\label{sec:conclusion}
In this paper, we proposed an approach to learn a neural network controller that is guaranteed to satisfy a given STL specification. We first provided a general procedure to construct a set of trainable HOCBFs from an STL formula. Then we applied BarrierNet to train these HOCBFs together with other parameters in the neural network to improve the robustness of the controller. Simulation results show that our approach converges within fewer iterations and achieves a higher robustness score than our previous approach without BarrierNet. Future work includes the extension to full STL specifications (rather than fragments) and predicates other than reaching and avoiding circular regions.






\bibliographystyle{IEEEtran}
\bibliography{references}

\end{document}